\documentclass[11pt,a4paper]{article}

\usepackage[top=30truemm,bottom=30truemm,left=30truemm,right=30truemm]{geometry}

\usepackage[leqno]{amsmath}
\usepackage{amsthm} 
\usepackage{amssymb}

\usepackage[section]{algorithm}
\usepackage{algorithmic}
\usepackage{graphicx,color}

\newtheorem{thm}{Theorem}
\newtheorem{lem}{Lemma}

\renewcommand{\chi}{\delta}

\def\supp{\mbox{supp}}

\def\zeros{\mathbf{0}}

\def\xb{\mbox{\boldmath $x$}}
\def\yb{\mbox{\boldmath $y$}}
\def\sb{\mbox{\boldmath $s$}}
\def\ob{\mbox{\boldmath $o$}}
\newcommand{\Del}[3]{\Delta_{{#1},{#2}} f({#3})}

\def\Fl{\mathcal{F}}
\def\Bl{\mathcal{B}}

\def\MM{{\bf (M2)}}
\def\MMM{{\bf (M3)}}

\def\epre{e_{\text{last}}}
\def\val{\text{Value}}

\def\EO{\text{EO}}
\def\MO{\text{MO}}

\def\R{\mathbb{R}}
\def\N{\mathbb{N}}

\DeclareMathOperator*{\maximize}{maximize}

\newcounter{counter}

\let\oldref\ref
\def\ref#1{{\normalfont\oldref{#1}}}
\def\eqref#1{{\normalfont(\oldref{#1})}}

\mathcode`@="8000 
{\catcode`\@=\active\gdef@{\mkern1mu}}

\mathchardef\Gamma="7100 \mathchardef\Delta="7101
\mathchardef\Theta="7102 \mathchardef\Lambda="7103
\mathchardef\Psi="7104 \mathchardef\Pi="7105 \mathchardef\Sigma="7106
\mathchardef\Upsilon="7107 \mathchardef\Phi="7108
\mathchardef\Psi="7109 \mathchardef\Omega="710A

\begin{document}

\title{On maximizing a monotone $k$-submodular function \\subject to a matroid constraint}
	\author{Shinsaku Sakaue\thanks{%
			 NTT Communication Science Laboratories, 2-4 Hikaridai, Seika-cho, Soraku-gun, Kyoto 619-0237, Japan 
		(sakaue.shinsaku@lab.ntt.co.jp). }}
	\date{\today 
	}

\maketitle

\begin{abstract}
\noindent
A $k$-submodular function is an extension of a submodular function in that its input is given by $k$ disjoint subsets instead of a single subset. 
For unconstrained nonnegative $k$-submodular maximization, 
Ward and {\v Z}ivn{\' y} proposed a constant-factor approximation algorithm, 
which was improved by the recent work of Iwata, Tanigawa and Yoshida 
presenting a $1/2$-approximation algorithm. 
Iwata {\it et al.} also provided a $k/(2k-1)$-approximation algorithm for 
monotone $k$-submodular maximization and proved that its approximation ratio is asymptotically tight. 
More recently, Ohsaka and Yoshida proposed constant-factor algorithms for monotone $k$-submodular maximization with 
several size constraints. 
However, while submodular maximization with various constraints has been extensively studied, 
no approximation algorithm has been developed for constrained $k$-submodular maximization, 
except for the case of size constraints. 

In this paper, we prove that a greedy algorithm outputs a $1/2$-approximate solution for monotone $k$-submodular maximization with a matroid constraint. 
The algorithm runs in $O(M|E|(\MO + k\EO))$ time,  
where $M$ is the size of a maximal optimal solution, 
$|E|$ is the size of the ground set, and 
$\MO, \EO$ represent the time for the membership oracle of the matroid and the evaluation oracle of the $k$-submodular function, 
respectively.

\end{abstract}

%
%
	
	\pagestyle{myheadings} 
	\thispagestyle{plain}
	\markboth{SHINSAKU SAKAUE} 
	{\uppercase{$k$-submodular maximization with a matroid constraint}}

\section{Introduction} 
Let $E$ be a finite set and $2^E$ be the family of all subsets in $E$. 
A function $f:2^E\to\R$ is called {\it submodular} if it satisfies
\[
f(X)+f(Y)\ge f(X\cup Y) + f(X\cap Y) 
\]
for all pairs of $X,Y\in 2^E$.  
It is well known that the following {\it diminishing return property} characterizes 
the submodular function:
\[
f(X\cup \{e\}) - f(X) \ge f(Y \cup \{e\}) - f(Y)
\] 
for any $X\subseteq Y$ and $e\in E\backslash Y$. 
The diminishing return property often appears in practice, 
and so various problems can be formulated as submodular function maximization (e.g.,  
sensor placement~\cite{krause2008robust,krause2008near}, 
feature selection~\cite{ko1995exact}, and document summarization~\cite{lin2010multi}). 
Unfortunately, submodular function maximization is known to be NP-hard.  
Therefore, approximation algorithms that can run in polynomial time 
have been extensively studied for submodular function maximization,   
some of which consider various constraints (e.g.,~\cite{buchbinder2015tight,calinescu2007maximizing,nemhauser1978analysis,sviridenko2004note}).

Recently, Huber and Kolmogorov~\cite{huber2012towards} proposed 
{$k$-submodular functions}, 
which express the submodularity on choosing $k$ disjoint sets of elements, instead of 
a single set.  
More precisely, let $(k+1)^E:=\{(X_1,\dots,X_k)\mid X_i\subseteq E\ (i=1,\dots,k),\ X_i\cap X_j=\emptyset\ (i\neq j)\}$. 
Then, a function $f:(k+1)^E\to\R$ is called {\it $k$-submodular} if,  
for any $\xb=(X_1,\dots,X_k)$ and $\yb=(Y_1,\dots,Y_k)$ in $(k+1)^E$, 
we have
\[
f(\xb)+f(\yb)\ge f(\xb\sqcup\yb) + f(\xb\sqcap\yb)
\]
where 
\begin{align*}
\xb\sqcap\yb&:=(X_1\cap Y_1,\dots,X_k\cap Y_k), \\
\xb\sqcup\yb&:=\biggl( X_1\cup Y_1\backslash \Bigl(\bigcup_{i\neq1} X_i\cup Y_i \Bigr),\dots, X_k\cup Y_k\backslash \Bigl(\bigcup_{i\neq k} X_i\cup Y_i \Bigr) \biggr). 
\end{align*} 
For an input $\xb=(X_1,\dots,X_k)$ of a $k$-submodular function, 
we define the size of $\xb$ by $\big|\bigcup_{i\in \{1,\dots,k\}} X_i\big|$. 
We say $f$ is {\it monotone} if $f(\xb)\le f(\yb)$ holds 
for any $\xb=(X_1,\dots, X_k)$ and $\yb=(Y_1,\dots,Y_k)$ with 
$X_i\subseteq Y_i$ for $i=1,\dots,k$. 
It is known that $k$-submodular functions arise as relaxation of NP-hard problems.    
$k$-submodular functions also appear in many applications.  
Therefore, the $k$-submodular function is recently a popular subject of study~\cite{gridchyn2013potts,hirai2015k}. 
If $k=1$, the above definition is equivalent to that of submodular functions. 
If $k=2$, the $k$-submodular function is equivalent to the so-called {\it bisubmodular function}, 
for which maximization algorithms have been widely studied~\cite{iwata2013bisubmodular,ward2014maximizing}. 
For unconstrained nonnegative $k$-submodular maximization, 
Ward and {\v Z}ivn{\' y}~\cite{ward2014maximizing} proposed a $\max\{1/3, 1/(1+a)\}$-approximation algorithm, 
where $a=\max\{1,\sqrt{(k-1)/4}\}$. 
Iwata, Tanigawa and Yoshida~\cite{iwata2016improved} improved the approximation ratio  to $1/2$. 
They also proposed a $k/(2k-1)$-approximation algorithm for monotone $k$-submodular maximization,  
and proved that, for any $\varepsilon>0$, 
a $((k+1)/2k+\varepsilon)$-approximation algorithm for maximizing monotone $k$-submodular functions 
requires exponentially many queries. This means their approximation ratio is asymptotically tight. 
More recently, Ohsaka and Yoshida~\cite{ohsaka2015monotone} proposed 
a $1/2$-approximation algorithm for monotone $k$-submodular maximization with a total size constraint 
(i.e., $\big| \bigcup_{i\in\{1,\dots,k\}} X_i\big|\le N$ for a nonnegative integer $N$) 
and a $1/3$-approximation algorithm for that with individual size constraints
(i.e., $|X_i|\le N_i$ for $i=1,\dots,k$ with associated nonnegative integers $N_1,\dots,N_k$).

In this paper,  
we prove that $1/2$-approximation can be achieved for monotone $k$-submodular maximization with a matroid constraint. 
This approximation ratio is asymptotically tight due to the aforementioned hardness result by Iwata {\it et al.}~\cite{iwata2016improved}. 
Given $\Fl\subseteq2^E$,  
we say a system $(E,\Fl)$ is {\it matroid} if the following holds:
\begin{description}
	\item[(M1)] $\emptyset\in\Fl$,
	\item[(M2)] If $A\subseteq B\in\Fl$ then $A\in\Fl$,
	\item[(M3)] If $A,B\in\Fl$ and $|A|<|B|$ then there exists $e\in B\backslash A$ such that $A\cup\{e\}\in\Fl$.
\end{description}
The elements of $\Fl$ are called {\it independent}, 
and we say $A\in\Fl$ is {\it maximal} if no $B\in\Fl$ satisfies $A\subsetneq B$.
Matroids include various systems; the total size constraint can be written as a special case of a matroid constraint.  
For example, the following systems $(E,\Fl)$ are matroids:
\begin{description}
	\item[(a)] $E$ is a finite set, and $\Fl:=\{F\subseteq E\mid |F|\le N\}$ where $N$ is a nonnegative integer. 
	\item[(b)] $E$ is the set of columns of a matrix over some field, and \\
	$\Fl:=\{F\subseteq E\mid \mbox{The columns in $F$ are linearly independent over the field}\}$.  
	\item[(c)] $E$ is the set of edges of a undirected graph $G$ with a vertex set $V$, and \\
	$\Fl:=\{ F\subseteq E\mid \mbox{The graph $(V,F)$ is a forest} \}$.
	\item[(d)] $E$ is a finite set partitioned into $\ell$ sets $E_1,\dots,E_\ell$ with associated nonnegative integers $N_1,\dots,N_\ell$, and 
	$\Fl:=\{ F\subseteq E\mid |F\cap E_i|\le N_i \mbox{ for $i=1,\dots,\ell$}\}$.
\end{description}
The total size constraint corresponds to {\bf (a)}, 
which is called a {\it uniform matroid}. 
Since submodular functions and matroids are capable of modeling various problems, 
approximation algorithms for submodular function maximization (i.e., $k=1$) with a matroid constraint have been extensively studied~\cite{calinescu2007maximizing,filmus2012tight,fisher1978analysis,golovin2011adaptive,lee2010maximizing}. 
However, to the best of our knowledge, no approximation algorithm has been studied for $k$-submodular maximization with a matroid constraint. 
Therefore, 
we show that a greedy algorithm provides a $1/2$-approximate solution for 
the following monotone $k$-submodular maximization with a matroid constraint:
\begin{align}
\maximize_{\xb\in(k+1)^E} \ f(\xb)
\qquad \mbox{subject to }\ 
\bigcup_{\ell\in\{1,\dots,k\}} X_\ell\in\Fl, \label{problem:ksubmomatroid}
\end{align}
where $\xb=(X_1,\dots,X_k)$. 
We also show that our algorithm incurs $O(M|E|(\MO + k\EO))$ computation cost,   
where $M$ is the size of a maximal optimal solution, and 
$\MO, \EO$ represent the time for the membership oracle of the matroid and the evaluation oracle of the $k$-submodular function, 
respectively. 
We see in Section~\ref{section:preliminaries} that all maximal optimal solutions for problem~\eqref{problem:ksubmomatroid} have equal size, 
which we denote by $M$ throughout this paper.

The rest of this paper is organized as follows. 
Section~\ref{section:preliminaries} reviews some basics of $k$-submodular functions and matroids. 
Section~\ref{section:maximize} discusses a greedy algorithm for problem~\eqref{problem:ksubmomatroid} and proves the $1/2$-approximation. 
We conclude this paper in Section~\ref{section:conclusion}.

\section{Preliminaries}\label{section:preliminaries} 
We elucidate some properties of a $k$-submodular function $f$ where $k\in\N$.  
Let $[k]:=\{1,2,\dots,k\}$. 
For $\xb=(X_1,\dots,X_k)$ and $\yb=(Y_1,\dots,Y_k)$ in $(k+1)^E$, 
we define a partial order $\preceq$ such that $\xb\preceq\yb$ if $X_i\subseteq Y_i$ for all $i\in[k]$. 
For $\xb,\yb\in(k+1)^E$ satisfying $\xb\preceq\yb$, we use $\xb\prec\yb$ if $X_i\subsetneq Y_i$ holds for some $i\in[k]$.
We also define 
\[
\Del{e}{i}{\xb}:=f(X_1,\dots,X_i\cup\{e\},\dots,X_k)-f(X_1,\dots,X_k)  
\]
for $\xb\in(k+1)^E$, $e\notin \bigcup_{\ell\in[k]}X_\ell$ and $i\in[k]$, 
which is a marginal gain when adding $e\in E$ to the $i$-th set of $\xb\in(k+1)^E$.
It is not hard to see that the
$k$-submodularity implies the {\it orthant submodularity}~\cite{ward2014maximizing}:
\[
\Del{e}{i}{\xb}\ge\Del{e}{i}{\yb}
\]
for any $\xb,\yb\in(k+1)^E$ with $\xb\preceq\yb,\ e\notin\bigcup_{j\in[k]} Y_j$, 
and $i\in[k]$, 
and the {\it pairwise monotonicity}:  
\[
\Del{e}{i}{\xb}+\Del{e}{j}{\xb}\ge0
\]
for any $\xb\in(k+1)^E$, $e\notin\bigcup_{\ell\in[k]} X_\ell$, 
and $i,j\in[k]$ with $i\neq j$. 
Actually, these properties characterize $k$-submodular functions: 
\begin{thm}[Ward and {\v Z}ivn{\' y}~\cite{ward2014maximizing}]
A function $f:(k+1)^E\to\R$ is $k$-submodular if and only if $f$ is orthant submodular and pairwise monotone.
\end{thm}
For notational ease, 
we identify $(k+1)^E$ with $\{0,1,\dots,k\}^E$, that is, 
we associate $(X_1,\dots, X_k)\in(k+1)^E$ with $\xb\in\{0,1,\dots, k\}^E$ 
by $X_i=\{e\in E\mid \xb(e)=i\}$ for $i\in[k]$. 
We sometimes abuse the notation, 
and simply write $\xb=(X_1\dots,X_k)$ by 
regarding a vector $\xb$ as disjoint $k$ subsets of $E$. 
For $\xb\in\{0,1,\dots,k\}^E$, we define $\supp(\xb):=\{e\in E\mid \xb(e)\neq0\}$; 
the size of $\xb$ can be written as $|\supp(\xb)|$. 
Let $\zeros$ be the zero vector in $\{0,1,\dots,k\}^E$. 
In what follows, 
we assume that the monotone $k$-submodular function $f$ in problem~\eqref{problem:ksubmomatroid} satisfies $f(\zeros)=0$ 
without loss of generality; 
if $f(\zeros)\neq0$, we redefine $f(\xb):=f(\xb)-f(\zeros)$ where $\xb\in(k+1)^E$. 

We now turn to some properties of matroid $(E,\Fl)$. 
An independent set $A\in\Fl$ is called a {\it bases} if it is a maximal independent set. 
We denote the set of all bases by $\Bl$.
It is known that each element in $\Bl$ has the same size (see, e.g.,~\cite[Theorem~13.5]{korte2012combinatorial}); 
the size is denoted by $M$ throughout this paper. 
Thus, we have the following lemma for the size of the maximal optimal solutions for problem~\eqref{problem:ksubmomatroid}. 
\begin{lem}\label{lemma:cardinality} 
The size of any maximal optimal solution for problem~\eqref{problem:ksubmomatroid} 
is $M$.
\end{lem}
\begin{proof}
Assume there is a maximal optimal solution $\ob$ such that $|\supp(\ob)|<M$. 
Let $\xb\in(k+1)^E$ be an arbitrary vector such that $\supp(\xb)\in\Bl$. 
Then, by \MMM, there exists $e\in\supp(\xb)\backslash\supp(\ob)$ such that 
$\supp(\ob)\cup\{e\}\in\Fl$. Since $f$ is monotone, 
by assigning arbitrary $i\in[k]$ to $\ob(e)$, we get $\Del{e}{i}{\ob}\ge0$; 
more precisely, $\Del{e}{i}{\ob}=0$ since $\ob$ is an optimal solution.  
This contradicts to the assumption that $\ob$ is a maximal  optimal solution.
\end{proof}

We also introduce the following lemma for later use. 
\begin{lem}\label{lemma:matroidadd} 
Suppose $A\in\Fl$ and $B\in\Bl$ satisfy $A\subsetneq B$. Then, for any $e\notin A$ satisfying $A\cup\{e\}\in\Fl$, there exists $e'\in B\backslash A$ such that $\{B\backslash\{e'\}\}\cup\{e\}\in\Bl$.		
\end{lem}

\begin{proof}
If $|B|-|A|=1$, by defining $e'=B\backslash A$, we get $\{B\backslash\{e'\}\}\cup\{e\}=A\cup\{e\}\in\Fl$. 
Since $|A\cup\{e\}|=|B|$, we have $\{B\backslash\{e'\}\}\cup\{e\}\in\Bl$. 
	
If $|B|-|A|\ge2$, then $|A\cup\{e\}|<|B|$. 
Thus, by applying \MMM\ iteratively, we can obtain $|B|-|A|-1$ elements $e_1,\dots,e_{|B|-|A|-1}\in B\backslash \{A\cup\{e \}\}$ such that 
\[
\{A\cup\{e\}\}\cup \{e_1\}\cup\dots\cup \{e_{|B|-|A|-1}\}\in\Bl.
\]
Therefore, defining $e'=B\backslash\{A\cup \{e_1\}\cup\dots\cup \{e_{|B|-|A|-1}\}\}$, we get 
\[
\{B\backslash\{e'\}\}\cup\{e\}=\{A\cup\{e\}\}\cup \{e_1\}\cup\dots\cup \{e_{|B|-|A|-1}\}\in\Bl.
\]
This completes the proof.
\end{proof}

\section{Maximizing a monotone $k$-submodular function with a matroid constraint}\label{section:maximize}
We present a greedy algorithm for problem~\eqref{problem:ksubmomatroid}; it runs in $O(M|E|(\MO + k\EO))$ time 
where $\MO$ and $\EO$ stand for the time for the membership oracle of matroid and the evaluation oracle of $k$-submodular function, 
respectively. 
We then prove that the greedy algorithm outputs a $1/2$-approximate solution for problem~\eqref{problem:ksubmomatroid}. 
In summary, this section proves the following theorem:

\begin{thm}
	For problem~\eqref{problem:ksubmomatroid}, a $1/2$-approximate solution can be obtained in $O(M|E|(\MO + k\EO))$ time.
\end{thm}

\subsection{Greedy algorithm and its complexity analysis}
We consider applying Algorithm~\ref{alg:greedy} to problem~\eqref{problem:ksubmomatroid}. 
\begin{algorithm}
	\caption{A greedy algorithm for $k$-submodular maximization with a matroid constraint}  \label{alg:greedy}
	\begin{algorithmic}[1]
		\REQUIRE a monotone $k$-submodular function $f:(k+1)^E\to\R$ and a matroid $(E,\Fl)$.
		\ENSURE a vector $\sb$ satisfying $\supp(\sb)\in \Bl$.
		\STATE $\sb\gets\zeros$.
		\FOR{$j=1$ to $M$}
		\STATE $\epre\gets\emptyset,\ \val\gets0$.
		\FOR{{\bf each} $e\in E\backslash\supp(\sb)$ such that $\supp(\sb)\cup\{e\}\in\Fl$}
		\STATE $i\gets\arg\max_{i\in[k]} \Del{e}{i}{\sb}$ 
		\IF{$\Del{e}{i}{\sb}\ge\val$}
		\STATE $\sb(\epre)\gets0$ unless $\epre=\emptyset$.
		\STATE $\sb(e)\gets i$.
		\STATE $\epre\gets e$ and $\val\gets\Del{e}{i}{\sb}$.
		\ENDIF
		\ENDFOR
		\ENDFOR
		\RETURN $\sb$.
		\end{algorithmic}
		\end{algorithm} 
First, we make a remark on using Algorithm~\ref{alg:greedy} in practice. 
In Step~2, the algorithm requires the value of $M$, the size of a maximal independent set.  
However, in practice, we need not calculate the value of $M$ beforehand. 
Instead, we continue the iteration while there exists $e\in E\backslash\supp(\sb)$ satisfying $\supp(\sb)\cup\{e\}\in\Fl$, 
which we check in Step~4.  
We can confirm that this modification does not change the output as follows. 
As long as $|\supp(\sb)|<M$, 
exactly one element is added to $\supp(\sb)$ at each iteration  
due to the monotonicity and \MMM, 
and, if $|\supp(\sb)|=M$, the iteration stops since $\supp(\sb)$ is a maximal independent set.  
Algorithm~\ref{alg:greedy} is described using $M$ to make it easy to understand the subsequent discussions. 
Note that, defining $\sb^{(j)}$ as the solution obtained after the $j$-th iteration, 
we have $|\supp(\sb^{(j)})|=j$ for $j\in[M]$.

We now examine the time complexity of Algorithm~\ref{alg:greedy}. 
Let $\EO$ be the time for the evaluation oracle of the $k$-submodular function $f$, 
and $\MO$ be the time for the membership oracle of the matroid $(E,\Fl)$. 
At the $j$-th iteration, 
the membership oracle is used at most $|E|$ times in Step~4, and 
the evaluation oracle is used  at most $k|E|$ times in Step~5. 
Thus, the time complexity of Algorithm~\ref{alg:greedy} is given by $O(M|E|(\MO + k\EO))$.

\subsection{Proof for 1/2-approximation} 
We now prove that Algorithm~\ref{alg:greedy} gives a $1/2$-approximate solution for problem~\eqref{problem:ksubmomatroid}. 
To prove this, we define a sequence of vectors $\ob^{(0)}, \ob^{(1)}, \dots,\ob^{(M)}$ as 
in~\cite{iwata2016improved,ohsaka2015monotone,ward2014maximizing}.

Let $(e^{(j)},i^{(j)})$ be the pair chosen greedily at the $j$-th iteration, 
and $\sb^{(j)}$ be the solution after the $j$-th iteration; we let $\sb=\sb^{(M)}$, the output of Algorithm~\ref{alg:greedy}.  
We define $\sb^{(0)}:=\zeros$ and let $\ob$ be a maximal optimal solution.  
In what follows, we show how to construct a sequence of vectors $\ob^{(0)}=\ob, \ob^{(1)}, \dots,\ob^{(M-1)},\ob^{(M)}=\sb$ satisfying 
the following:
\begin{align}
& \sb^{(j)}\prec\ob^{(j)}\ \mbox{if}\ j=0,1,\dots,M-1,\ \mbox{and}\ \sb^{(j)}=\ob^{(j)}=\sb\ \mbox{if}\ j=M.   \label{property1} \\
& O^{(j)}\in\Bl \ \mbox{for}\ j=0,1,\dots,M. \label{property2}
\end{align} 
More specifically, we see how to obtain $\ob^{(j)}$ from $\ob^{(j-1)}$ satisfying~\eqref{property1} and~\eqref{property2}.  
Note that $\sb^{(0)}=\zeros$ and $\ob^{(0)}=\ob$ satisfy~\eqref{property1} and~\eqref{property2}.
We define $S^{(j)}:=\supp(\sb^{(j)})$, $O^{(j)}:=\supp(\ob^{(j)})$ for each $j\in[M]$. 

We now describe how to obtain $\ob^{(j)}$ from $\ob^{(j-1)}$, 
assuming that $\ob^{(j-1)}$ satisfies 
\begin{align*}
\sb^{(j-1)}\prec\ob^{(j-1)}, \ \mbox{and}\ O^{(j-1)}\in\Bl. 
\end{align*}
Since $\sb^{(j-1)}\prec\ob^{(j-1)}$ means $S^{(j-1)}\subsetneq O^{(j-1)}$, and
$e^{(j)}$ is chosen to satisfy $S^{(j-1)}\cup\{e^{(j)}\}\in\Fl$, we see from Lemma~\ref{lemma:matroidadd} that 
there exists $e'\in O^{(j-1)}\backslash S^{(j-1)}$ satisfying $\{O^{(j-1)}\backslash\{e'\}\}\cup\{e^{(j)}\}\in\Bl$.  
We let $o^{(j)}=e'$ and define $\ob^{(j-1/2)}$ as the vector obtained by assigning $0$ to the $o^{(j)}$-th element of $\ob^{(j-1)}$. 
We then define $\ob^{(j)}$ as the vector obtained from $\ob^{(j-1/2)}$ by assigning $i^{(j)}$ to the $e^{(j)}$-th element. 
The vector thus constructed, $\ob^{(j)}$, satisfies  
\begin{equation}\label{property1j}
O^{(j)}=\{O^{(j-1)}\backslash\{o^{(j)}\}\}\cup\{e^{(j)}\}\in\Bl.
\end{equation}
Furthermore, since $\ob^{(j-1/2)}$ satisfies 
\[
\sb^{(j-1)}\preceq\ob^{(j-1/2)}, 
\]
we have the following property for $\ob^{(j)}$: 
\begin{equation}\label{property2j}
\sb^{(j)}\prec\ob^{(j)}\ \mbox{if}\ j=1,\dots,M-1, \ \mbox{and}\ \sb^{(j)}=\ob^{(j)}=\sb\ \mbox{if}\ j=M, 
\end{equation} 
where the strictness of the inclusion for $j\in[M-1]$ can be easily confirmed from $|S^{(j)}|=j<M=|O^{(j)}|$. 
Thus, applying the above discussion for $j=1,\dots,M$ iteratively,  
we see from~\eqref{property1j} and~\eqref{property2j} that 
the obtained sequence of vectors $\ob^{(0)},\ob^{(1)},\dots,\ob^{(M)}$ satisfies~\eqref{property1} and~\eqref{property2}.

We now prove the following inequality for $j\in[M]$:
\begin{equation}\label{ineq:stepwise}
f(\sb^{(j)})-f(\sb^{(j-1)})\ge f(\ob^{(j-1)})-f(\ob^{(j)}).
\end{equation} 
Since $S^{(j-1)}\cup\{o^{(j)}\}\subseteq O^{(j-1)}\in\Bl$ holds for each $j\in[M]$, 
we get the following inclusion from \MM:
\begin{equation*}
S^{(j-1)}\cup\{o^{(j)}\}\in\Fl \label{eq:stepwisefeasible}
\end{equation*}
for any $j\in[M]$.  
Therefore, 
for the pair $(e^{(j)}, i^{(j)})$, which is chosen greedily, 
we have  
\begin{equation}\label{ineq1}
	\Del{e^{(j)}}{i^{(j)}}{\sb^{(j-1)}}\ge\Del{o^{(j)}}{\ob^{(j-1)}(o^{(j)})}{\sb^{(j-1)}}.  
\end{equation} 
Furthermore, since $\sb^{(j-1)}\preceq\ob^{(j-1/2)}$ holds, 
orthant submodularity implies 
\begin{equation}\label{ineq2}
\Del{o^{(j)}}{\ob^{(j-1)}(o^{j})}{\sb^{(j-1)}}\ge\Del{o^{(j)}}{\ob^{(j-1)}(o^{j})}{\ob^{(j-1/2)}}.
\end{equation}
Using~\eqref{ineq1} and~\eqref{ineq2}, we get
\begin{align*}
f(\sb^{(j)})-f(\sb^{(j-1)})=&\Del{e^{(j)}}{i^{(j)}}{\sb^{(j-1)}}\\
\ge&\Del{o^{(j)}}{\ob^{(j-1)}(o^{(j)})}{\sb^{(j-1)}} \\
\ge&\Del{o^{(j)}}{\ob^{(j-1)}(o^{j})}{\ob^{(j-1/2)}}\\
\ge&\Del{o^{(j)}}{\ob^{(j-1)}(o^{j})}{\ob^{(j-1/2)}}-\Del{e^{(j)}}{i^{(j)}}{\ob^{(j-1/2)}} \\
=& f(\ob^{(j-1)})-f(\ob^{(j)}),
\end{align*}
where the third inequality comes from the monotonicity, i.e., $\Del{e^{(j)}}{i^{(j)}}{\ob^{(j-1/2)}}\ge0$. 

By~\eqref{ineq:stepwise}, we have 
\[
f(\ob)-f(\sb)=\sum_{j=1}^{M}(f(\ob^{(j-1)})-f(\ob^{(j)}))\le\sum_{j=1}^{M}(f(\sb^{(j)})-f(\sb^{(j-1)}))=f(\sb)-f(\zeros)= f(\sb), 
\]
which means $f(\sb)\ge f(\ob)/2$. 

\section{Conclusions}\label{section:conclusion}
We proved that a $1/2$-approximate solution can be obtained  
for monotone $k$-submodular maximization with a matroid constraint via a greedy algorithm.  
Our approach follows the techniques shown in~\cite{iwata2016improved,ohsaka2015monotone,ward2014maximizing}. 
The proved approximation ratio is asymptotically tight due to the hardness result shown in~\cite{iwata2016improved}. 
We also showed that the proposed algorithm incurs $O(M|E|(\MO + k\EO))$ computation cost.

\bibliographystyle{plain}
\bibliography{./mybib}

\end{document}